\documentclass[12pt]{article}
\usepackage{latexsym,graphicx}
\usepackage{subfigure}
\usepackage{amssymb}
\usepackage{amsmath}
\usepackage{amscd}
\usepackage{amsthm}
\usepackage{float}
\usepackage[left=2cm,top=2.5cm,right=2.5cm,bottom=1.5cm]{geometry}
\theoremstyle{definition}
\newtheorem{definition}{Definition}    
\newtheorem{theorem}{Theorem}[section]      
\begin{document}
\begin{center}
\large{\bf{On the existence of deformations and dimensional reduction in wormholes and black holes}} \\
\vspace{5mm}
\normalsize{Nasr Ahmed$^a$ and H. Rafat$^b$,$^c$}\\
\vspace{5mm}
\small{\footnotesize$^{a}$ \textit{Astronomy Department, National Research Institute of Astronomy and Geophysics, Helwan, Cairo, Egypt}\footnote{nasr.ahmed@nriag.sci.eg}} \\
\small{\footnotesize$^{b}$\textit{Department of Mathematics, Faculty of Science, Tanta University, Tanta, Egypt.}\\$^c$\textit{Mathematics Department, Faculty of Science, Taibah University, Saudi Arabia}} \\
\vspace{2mm}
\end{center}  
\begin{abstract}
The deformation retract is, by definition, a homotopy between a retraction and the identity map. We show that applying this topological concept to Ricci-flat wormholes/black holes implies that such objects can get deformed and reduced to lower dimensions. The homotopy theory can provide a rigorous proof to the existence of black holes/wormholes deformations and explain the topological origin. The current work discusses such possible deformations and dimensional reductions from a global topological point of view, it also represents a new application of the homotopy theory and deformation retract in astrophysics and quantum gravity.

\end{abstract}
{\it Mathematical Subject Classification 2010}: 

\section{Introduction and motivation}

In general relativity, the theory of tidal deformation of black holes and neutron stars has been a subject of intensive studies. In case of no other sources of gravity, static black holes are described by a spherically symmetric Schwarzschild space-time. However, in realistic scenarios black holes can get distorted due to tidal forces when they are included in binary systems or surrounded by an accretion disks \cite{gurl}. These tidal effects of compact bodies might be accessible to measurement by the current generation of gravitational-wave detectors\cite{tid1,tid2,tid3,tid4,tid5}. Deformations in wormholes and black holes have been investigated in modified gravity theories using a geometric approach called the Anholonomic Frame Deformation Method (AFDM) developed in \cite{g1, g2,g3}. Deformation of extremal black holes from stringy interactions has been studied in \cite{g4} where analytic solutions for the linearized
metric deformations to near-horizon extremal Kerr space-times have been obtained. Another approach called Minimal Geometric Deformation (MGD) has been developed to study the exterior space-time of a self-gravitating system in the Brane-world scienario \cite{brr1}. This MGD approach has been extended in \cite{brr2} where a solution for the deformation undergone by the radial metric component
when time deformations are produced by bulk gravitons has been identified. It is also well known that stars can get deformed by magnetic fields \cite{de1,de2}. Due to the intense magnetic fields of neutron stars, particularly the subset known as magnetars, they possess ellipticities which are roughly proportional to the magnetic energy \cite{de3, de4}. Neutron star deformation due to multipolar magnetic fields has been discussed in \cite{de5}.\par

The two concepts of deformation and dimensional reduction are both included in the definition of deformation retract (section \ref{sec1}). So, in addition to the geometrical/physical deformation examples mentioned above, the current work also discusses the dimensional reduction from a topological point of view. Dimensional reduction plays a very important role in quantum gravity, Several approaches to quantum gravity suggest that the effective dimension of space-time
reduces to lower dimensions near the Planck scale. The spontaneous dimensional reduction associated with black holes evaporation has been discussed in \cite{vii}. It is generally believed that the effective dimensionality of space-time depend on the energy scale which has revived the interest in lower dimensional physics \cite{dim1,dim2,dim3,dim4,dim5}. While many authors believe that at very short distances the number of spatial dimensions increases (extra-dimensions) \cite{dim6,dim7}, it is also possible that the number of spatial
dimensions decreases as the Planck length is approached. The dimensional reduction has been studied in several contexts \cite{vii}. The two dimensions (2D) has become of particular interest after the appearance of a black hole in string theory \cite{bb1,bb2}. An alternative dimensional reduction scenario with fractal space-time has been suggested in \cite{bb16}. Some much simpler gravity theories have been formulated in (2+1)-dimensions \cite{bb3,bb4,bb5,bb6, bb7} and (1+1)-dimensions \cite{bb8, bb9, bb10, bb11, bb12}, where associated quantum theories are exactly solvable \cite{bb13}. A geometric dimensional reduction framework has been proposed in which space-time is a recursive lattice-network of lower-dimensional substructures each has a fundamental length scale \cite{dim2,bb14}. In addition to the natural solution to the hierarchy problem, this geometric dimensional reduction framework provides a range of phenomenological signatures that might be observable in future experiments. \par

Let's see how the current work is related to the previous works on geometrical/physical deformations and dimensional reduction of such astrophysical objects. First, all of these works studied deformations and dimensional reduction from a pure geometrical local side while the current study discusses the topological global side making use of the deformation retract definition. Secondly, since wormholes and black holes deformations still never been observed yet, a solid topological base for the existence of such deformations represents a great support to the theory. In other words, the current work explains the topological origin of such geometrical/physical deformations and provides a rigorous proof to its existence. Thirdly, the current work introduces a new application of the deformation retract in astrophysics and quantum gravity. We believe that the topological retraction theory is so powerful in explaining some modern topics in mathematical physics and simplifying calculations, it has already been used to provide a topological explanation to the holographic principle in \cite{holograph}.\par

Topology is the appropriate mathematical tool for the study of shapes and spaces (without a notion of distance) which can be continuously deformed into each other, continuous deformations mean twisting and stretching but not tearing or puncturing. For example, a sphere is topologically equivalent to a cube and a square is topologically equivalent to a circle \cite{kenna}. In general relativity, space-time exists as a manifold which opens the door to the use of topological concepts and methods. The topology change in general relativity has been discussed in \cite{topo2} and some applications of differential topology in general relativity has been mentioned in \cite{topo3}. In the late 1970s, it was shown that quantum field theory techniques can be used to obtain topological invariants of manifolds \cite{Witt,schwarz}. The topological quantum field theory was defined in \cite{ati}.\par

The deformations in topology are described by homoptopy, a very useful topological notion defined as a continuous deformation of one continuous function to another \cite{kinsey}. A nice mental picture of a homotopy is the human aging which is a continuous process where the topological shape of an infant is related to the shape of a wrinkled old person by a homotopy describing the shape at every age between \cite{kinsey}. The waving flag is a nother example, it is topologically equivalent to a rectangle at all times but the way it is embedded in the 3D space varies with time. Thus, homotopy is an equivalence relation between maps and it has been so useful in providing proofs and simplifying calculations. The homotopy theory has been applied to the study of defects in the ordered media of condensed matter physics \cite{cond}. The topological theory of defects in ordered media has been discussed in details in \cite{cond2} where is has been shown that all the paradoxes appear from the old theory can be resolved within the larger frameworks of homotopy theory. There is a particular class of defects which can not be removed by continuous deformation of fields called topological defects. Topological defects or 'Topological solitons' 
are topologically distinct (permanently stable) solutions to nonlinear partial differential equations, and homotopy theory is used to determine when the solutions are truly distinct. So, the term 'topological' comes from homotopical distinction between the solution which has these objects and the vacuum solution. Topological defects appear in both condensed matter physics and cosmology. Homotopy quantum field theory (HQFT) \cite{ho} represents a branch of topological quantum field theory (TQFT) \cite{ati} where the idea of a (TQFT) has been applied to maps from manifolds into topological spaces. A homotopy approach to quantum gravity has been described in \cite{hoo} where a finite-dimensional quantum theory has been constructed from general relativity by a homotopy method. A homotopy theory of algebraic quantum field theories has been developed in \cite{hooo}. Topological Concepts in Gauge Theories have been discussed in \cite{too}. The current work represents a new application of the homotopy theory in gravitational physics.
\par

This paper is organized as follows: Section 1 is an introduction. In section 2 we discuss the basic mathematical definitions of the topological concepts used in the paper. In sections 3, we review the retraction technique for Ricci-flat spaces introduced by the authors in previous studies. In section 4, we use the homotopy theory to prove two theorems on the existence of deformations and dimensional reductions for wormholes and black holes. We follow the proofs by a discussion on some basic aspects of the homotopy, and on the null energy condition violation in wormholes. The conclusion is included in section 5.

\section{Retraction, homotopy and deformation retract} \label{sec1}

\theoremstyle{definition}
\begin{definition}
A subspace $A$ of a topological space $X$ is called a retract of $X$, if there exists a continuous map $r : X \rightarrow A $ such that $X$ is open and $r(a) = a$ (identity map), $\forall a \in A$. Because the continuous map $r$ is an identity map from $X$ into $A \subset X$, it preserves the position of all points in $A$. 
\end{definition}

\theoremstyle{definition} \label{def2}
\begin{definition}{Homotopy:} \cite{geo} Let $X$, $Y$ be smooth manifolds and $f: X \rightarrow Y$ a smooth map between them. A homotopy or deformation of the map $f$ is a smooth map $F: X \times [0,1] \rightarrow Y$ with the property $F(x,0)=f(x)$. Each of the maps $f_t(x)=F(x,t),~t\in [0,1]$ is said to be homotopic to the initial map $f_0=f$ and the map of the whole cylinder $X \times [0,1]$ is called a homotopy.
Two continuous maps $f: X \rightarrow Y$ and $g: X \rightarrow Y$ are called homotopic if there exists a continuous map $F: X \times [0,1] \rightarrow Y$ with $F(x,0)=f(x)$ and $F(x,1)=g(x)$ $\forall x \in X$.

\end{definition}

\theoremstyle{definition}
\begin{definition}{Deformation retract:}\label{def1}
A subset $A$ of a topological space $X$ is said to be a deformation retract if there exists a
retraction $r: X \rightarrow A$, and a homotopy $f : X \times [0,1] \rightarrow X $ such that \cite{11}: $f (x, 0) = x$ $\forall x\in X$, $f(x, 1) = r (x)$ $\forall r\in X$, $f(a, t) = a$ $\forall a\in A, t \in [0,1]$.
\end{definition}

A retraction $r:S^2 \rightarrow \left\{(1,0,0)\right\} \in S^2$ taking the sphere to a point defined by $r(x,y,z)=(1,0,0)$ is a good example \cite{kinsey}. Retractions of Stein spaces has been studied in \cite{stien}. The retraction theory has been investigated in many branches of topology and differential geometry \cite{retra,stien,nash}. While homotopy theory has many applications in mathematical physics, most of the studies on the deformation retract theory -if not all- are pure mathematical studies. The retraction is called deformation retract if it preserves the essential shape of the space. For example, both the cylinder and the mobius band deformation retract to a circle and they are said to have the same homotopy type \cite{kinsey}. In a previous study \cite{holograph}, an application to the theory of topological retracts in mathematical physics has been introduced. It has been shown that the holographic principle can be represented mathematically by topological retraction of $n$-dimensional space (hologram) to $n-1$ subspace (a boundary) that preserves the position of all points (preserving all information describing the $n$-dimensional space). We start by proving the possibility of wormhole space deformations through introducing a homotopy $f$, i.e. continuous deformations of wormhole space $W$, $f: W \times [0,1] \rightarrow W$. Since the retraction of wormhole space $W$ has been given in \cite{holograph}, we will be able to discuss the deformation retract of $W$ with the existence of a homotopy $f$ . 

\section{Review of retraction} \label{ret}

Wormholes are hypothetical bridges between two regions of space-time connected by a throat, such objects could be described topologically as folding in the space-time fabric. Many interesting studies of Lorentzian wormholes have been described in \cite{wormhole} while four and higher-dimensional wormholes were introduced by Hawking and Coleman \cite{h1,h2}. A retraction method applicable only for Ricci-flat geometries has been discussed in \cite{holograph} and applied to the $5$D Ricci-flat wormhole space-time $W$ \cite{basic}:
\begin{equation}
ds^{2}=\frac{r^2-a^2}{r^2+a^2}(-dt^2+dz^2)+dr^2+(r^2+a^2)(d\theta^2+\sin^2\theta \,d\phi^2)-\frac{4ar}{r^2+a^2} dt\, dz. \label{mett}
\end{equation}
There are two asymptotic regions, corresponding to $r\rightarrow \pm \infty$, where the metric has the following form
\begin{equation}
ds^{2}=-dt^{2}+dz^{2}+dr^{2}+r^2 (d\theta^2+\sin^2\theta \,d\phi^2)
\end{equation}
This describes $(Minkowskian)_5$ if $z$ is a real line, or $(Minkowskian)_4 \times S^1$ if $z$ is a circle. After comparing the Ricci-flat metric (\ref{mett}) with the general form of the $5D$ flat metric $ds^{2}= -dx_{o}^{2}+ \sum_{i=1}^4 dx_{i}^{2}$ and making use of the basic metric definition $ds^2 = g_{ij} dx^i dx^j$, the following coordinate relations have been obtained
\begin{eqnarray} \label{xi}
x_o&=&\pm \sqrt{\frac{r^2-a^2}{r^2+a^2} t^2+C_o}, \  \  \  x_1=\pm \sqrt{\frac{r^2-a^2}{r^2+a^2} z^2+C_1}, \  \  \  x_2=\pm \sqrt{r^2+C_2},\\  \nonumber
x_3&=&\pm \sqrt{(r^2+a^2)\theta^2+C_3}, \  \  \  x_4=\pm \sqrt{(r^2+a^2) \sin^2\theta \, \phi^2+C_4}.
\end{eqnarray}
Where $C_{o}$, $C_{1}$,$C_{2}$,$C_{3}$ and $C_{4}$ are constants of integration. We have used these relations to study the geodesic retractions of the $5$D wormhole space-time $W$ (\ref{mett}). To find a geodesic which is a subset of $W$ we have used the Euler-Lagrange equations associated to the Lagrangian $L(x^{\mu},\dot{x}^{\mu})=\frac{1}{2}g_{\mu \nu}\dot{x}^{\mu}\dot{x}^{\nu}$. The following equations have been obtained
\begin{equation} \label{1}
-\frac{r^2-a^2}{r^2+a^2}\dot{t}=K. 
\end{equation}
\begin{equation}\label{2}
(r^2+a^2) \sin^2\theta \, \dot{\phi}=h.
\end{equation}
\begin{equation}\label{3}
\frac{d}{d\lambda}(\dot{r})-\left[-\frac{2a^2r}{(r^2+a^2)^2}\dot{t}^2+\frac{2a^2r}{(r^2+a^2)^2}\dot{z}^2+r(\dot{\theta}^2+\sin^2\theta \, \dot{\phi}^2)+2a\frac{r^2-a^2}{(r^2+a^2)^2}\dot{t}\dot{z}\right]=0.
\end{equation}
\begin{equation}\label{4}
\frac{d}{d\lambda}\dot{\theta}(r^2+a^2)-\left[\frac{1}{2}(r^2+a^2)\sin 2\theta \, \dot{\phi}^2\right]=0.
\end{equation}
\begin{equation}\label{6}
\frac{d}{d\lambda}\left[(r^2+a^2)\dot{\phi}\sin^2\theta\right]=0.
\end{equation}
\begin{equation}\label{7}
\frac{r^2-a^2}{r^2+a^2}\dot{t}+\frac{2ar}{r^2+a^2}\dot{z}=A.
\end{equation}
\begin{equation}\label{8}
\frac{r^2-a^2}{r^2+a^2}\dot{z}-\frac{2ar}{r^2+a^2}\dot{t}=B.
\end{equation}
Where $K$, $h$, $A$ and $B$ are constants. Equations of motion for $\theta$ and $\phi$ admit the solution $\theta=\frac{\pi}{2} , \, \dot{\phi}=\frac{D}{r^2+a^2}$ where $D$ is the integration constant describing the orbital angular momentum of the geodesic particle. We have used these equations to see what types of geodesic retractions we get. From (\ref{1}), if $K=0$ this implies that $t=C$ where $C$ is a constant, or $r=a$. For $C=0$, the coordinates (\ref{xi}) become
\begin{eqnarray} \label{xi1}
x_o^{C=0}&=&\pm \sqrt{C_o}, \  \  \  x_1^{C=0}=\pm \sqrt{\frac{r^2-a^2}{r^2+a^2} z^2+C_1}, \  \  \  x_2^{C=0}=\pm \sqrt{r^2+C_2},\\  \nonumber
x_3^{C=0}&=&\pm \sqrt{(r^2+a^2)\theta^2+C_3}, \  \  \  x_4^{C=0}=\pm \sqrt{(r^2+a^2) \sin^2\theta \, \phi^2+C_4}.
\end{eqnarray}
$ds^2=x_{1}^{2}+x_{2}^{2}+x_{3}^{2}-x_{o}^{2}>0$ is a circle $S_1 \subset W$, this geodesic is a retraction in the 5D space-time $W$ represented by the metric (\ref{mett}). For $r=a$, the coordinates (\ref{xi}) become
\begin{eqnarray} \label{xi2}
x_o^{r=a}&=&\pm \sqrt{C_o}, \  \  \  x_1^{r=a}=\pm \sqrt{C_1}, \  \  \  x_2^{r=a}=\pm \sqrt{r^2+C_2},\\  \nonumber x_3^{r=a}&=&\pm \sqrt{2a^2\theta^2+C_3}, \  \  \  x_4^{r=a}=\pm \sqrt{2a^2 \sin^2\theta \, \phi^2+C_4}.
\end{eqnarray}
$ds^2=x_{1}^{2}+x_{2}^{2}+x_{3}^{2}-x_{o}^{2}>0$ is a circle $S_2 \subset W$, this geodesic is a retraction in $W$. From (\ref{2}), if $h=0$ this implies that $\phi=H$ where $H$ is a constant, or $\theta=0$. For $H=0$, the coordinates (\ref{xi}) become
\begin{eqnarray} \label{xi3}
x_o^{H=0}&=&\pm \sqrt{\frac{r^2-a^2}{r^2+a^2} t^2+C_o}, \  \  \  x_1^{H=0}=\pm \sqrt{\frac{r^2-a^2}{r^2+a^2} z^2+C_1}, \  \  \  x_2^{H=0}=\pm \sqrt{r^2+C_2},\\  \nonumber
x_3^{H=0}&=&\pm \sqrt{(r^2+a^2)\theta^2+C_3}, \  \  \  x_4^{H=0}=\pm \sqrt{C_4}.
\end{eqnarray}
$ds^2=x_{1}^{2}+x_{2}^{2}+x_{3}^{2}-x_{o}^{2}>0$ is a circle $S_3 \subset W$, this geodesic is a retraction in $W$. Finally, For $\theta=0$, we get
\begin{eqnarray} \label{xi4}
x_o^{\theta=0}&=&\pm \sqrt{\frac{r^2-a^2}{r^2+a^2} t^2+C_o}, \  \  \  x_1^{\theta=0}=\pm \sqrt{\frac{r^2-a^2}{r^2+a^2} z^2+C_1}, \  \  \  x_2^{\theta=0}=\pm \sqrt{r^2+C_2},\\  \nonumber
x_3^{\theta=0}&=&\pm \sqrt{C_3}, \  \  \  x_4^{\theta=0}=\pm \sqrt{C_4}.
\end{eqnarray}
$ds^2=x_{1}^{2}+x_{2}^{2}+x_{3}^{2}-x_{o}^{2}>0$ is a circle $S_4 \subset W$, this geodesic is a retraction in $W$. So, the retraction of $W$ can be defined as $R : W \rightarrow S_i ,\,i=1,2,3,4$ which means that: some types of the geodesic retractions of the 5D wormhole space-time $W$ are circles $S_i \subset W$. 
\section{Space-time deformation}

In this section, we study the possibility of continuous deformations of the 5D Ricci-flat wormhole space-time $W$, the $nD$ Schwarzchild blackhole space-time $Sc$ and their possible retraction into subspaces. To prove the possible deformations with dimensional reductions of Ricci-flat wormholes, the following theorem needs to be proved:
\begin{theorem} \label{th1}
5D Ricci-flat wormhole space $W$ can get continuously deformed into other spaces and reduced into subspaces. 
\end{theorem}

\begin{proof}

Recalling definition \ref{def1}, in addition to the retraction $R: W \rightarrow S_i$ we need to define a homotopy between the retraction and the identity map on $W$ ( $f : W \times [0,1] \rightarrow W $ ) such that $f (s, 0) = s$ $\forall s\in W$, $f(s, 1) = R (s)$ $\forall R\in W$, $f(a, t) = a$ $\forall a\in S_i, t \in [0,1]$ to have a deformation retract on $W$. For the retraction of $W$ into a geodesic $S_1 \subset W$, a homotopy can be defined as $\eta^{S_1} : W \times [0,1] \rightarrow W$ where
\begin{equation} \label{hom1}
\eta^{S_1}(m,s)=(1-s)~\eta(m,0)+s~\eta^{S_1}(m,1),~~~~~~~~~~~~\forall m \in W~ \&~ \forall s \in [0,1].
\end{equation}
With
\begin{equation}
\eta(m,0)=\left\{x_o, x_1, x_2, x_3, x_4\right\}~~\& ~~\eta^{S_1}(m,1)=\left\{x_o^{C=0}, x_1^{C=0}, x_2^{C=0}, x_3^{C=0}, x_4^{C=0}\right\}.
\end{equation}
And the coordinates $x_i^{C=0},~i=0,1,2,3,4$ are given by (\ref{xi1}) . The homotopies of the retraction of $W$ into geodesics $S_2,~S_3,S_4 \subset W$, are defined respectively as
\begin{equation}\label{hom}
\eta^{S_2}(m,s)=(1-s)~\eta(m,0)+s~\eta^{S_2}(m,1),~~~~~~~~~~~~~~~~~~\forall m \in W~ \&~ \forall s \in [0,1].
\end{equation}
\begin{equation} \label{hom2}
\eta^{S_3}(m,s)=\cos \frac{\pi s}{2}~\eta(m,0)+\sin \frac{\pi s}{2}~\eta^{S_3}(m,1),~~~~~~~~~~~~\forall m \in W~ \&~ \forall s \in [0,1].
\end{equation}
\begin{equation} \label{hom3}
\eta^{S_4}(m,s)=\cos \frac{\pi s}{2}~\eta(m,0)+\sin \frac{\pi s}{2}~\eta^{S_4}(m,1),~~~~~~~~~~~~\forall m \in W~ \&~ \forall s \in [0,1].
\end{equation}
With 
\begin{eqnarray*}
\eta^{S_2}(m,1)&=&\left\{x_o^{r=a}, x_1^{r=a}, x_2^{r=a}, x_3^{r=a}, x_4^{r=a}\right\}\\ .
\eta^{S_3}(m,1)&=&\left\{x_o^{H=0}, x_1^{H=0}, x_2^{H=0}, x_3^{H=0}, x_4^{H=0}\right\}\\ .
\eta^{S_4}(m,1)&=&\left\{x_o^{\theta=0}, x_1^{\theta=0}, x_2^{\theta=0}, x_3^{\theta=0}, x_4^{\theta=0}\right\} .
\end{eqnarray*}
And the coordinates $x_i^{r=a}$, $x_i^{H=0}$, $x_i^{\theta=0} ~i=0,1,2,3,4$ are given respectively by (\ref{xi2}), (\ref{xi3}) and (\ref{xi4}). With the homotopies defined here, and the retraction of $W$ defined in section (\ref{ret}) we now have a deformation retract defined on $W$. Defining such homotopies which satisfying the three conditions in definition 3 proves the possible deformations of $W$.
\end{proof}

Now, some basic mathematical aspects of homotopies such as existence, unicity/classes of transformations and well-definiteness should be discussed carefully. Since all the geodesic retractions we found are circles $S_i \in W$, a homotopy $\eta^{S_i} : W \times [0,1] \rightarrow W$ is unique if it has the same form for any $i$. But we can see that this is not the case as the two homotopies (\ref{hom1}) for $i=1$ and (\ref{hom2}) for $i=3$ are different. What we have done is that we had to replace $(1-s)$ and $s$ in (\ref{hom1}) and $(\ref{hom})$ with $\cos \frac{\pi s}{2}$ and $\sin \frac{\pi s}{2}$ in $(\ref{hom2})$ and $(\ref{hom3})$ so that they all satisfy the same definition of the deformation retract (\ref{def2}). Then, this proves the non-unicity of our defined homotopy between the retraction and the identity map on $W$. This can also be directly deduced from the fact that the homotopy (\ref{hom1}) has been defined so that it satisfies the deformation retract definition (\ref{def1}), and a space can deformation retract into different subspaces. We also recall that such a homotopy $\eta^{S_i} : W \times [0,1] \rightarrow W$ is well defined by definition (\ref{def2}) and does exist under the three mathematical conditions in this definition. Any other possible homotopies will lead to similar results as they must satisfy the conditions in definition (\ref{def1}).\par

The physical interpretation of theorem (\ref{th1}) is that such wormholes can get continuously deformed, whether due to tidal effects, magnetic fields or any other physical effects, and reduced to lower dimensions. The existence of wormhole solutions essentially depend on some special properties of the matter source term in the Einstein field equations and on the violation of the null energy condition (NEC) $T_{\mu \nu}k^{\mu}k^{\nu}<0$, where $T_{\mu \nu}$ is the matter energy-momentum tensor and $k^{\mu}$ is any null vector. In terms of the energy density $\rho$ and pressure $p$, the NEC viloation is written as $\rho + p<0$ which requires the existence of some forms of exotic matter. Because deformations discussed here are homotopic, it should not affect the nature of the matter. The exotic matter will be still there after the deformations which means the NEC vilation is not affected by such topological transformations or retractions. A good example is that if we keep deforming a piece of plastic, this deformations will never turn it into glass or any other different matter. Another point is that because the deformation retract preserves the essential shape of the space, there is a guarantee that the new geometry will still correspond to a wormhole in the sense that both of them (the deformed and original) can be continuously deformed into each others (homotopy is an equivalence relation). It is also important to mention that while the violation of the null energy condition implies the violations of all other (weak $\rho \geq 0$, $\rho + p \geq 0$, strong $\rho + 3p \geq 0$ and dominant $\rho \geq \left|p\right|$) energy conditions, it has been shown that those classical linear energy conditions should be replaced by other nonlinear energy conditions in the presence of semiclassical quantum effects \cite{ec1}. Although stayed undoubted for a long time, it has been suggested that these classical energy conditions are not fundamental physics  and can not be valid in completely general situations \cite{ec2}.\par
The same topological analysis can provide a solid mathematical base for the existence of physical deformations in Ricci-flat black holes such as $nD$ Schwarzchild static blackholes and Kerr rotating black hole. For example, for $5D$ Ricci-flat Schwarzchild static blackholes the following theorem can be proved
\begin{theorem} \label{th2}
$5D$ Ricci-flat black hole spaces can get continuously deformed into other spaces and reduced into subspaces.
\end{theorem}
\begin{proof}
The Schwarzchild metric in $(n+1)$ dimensions is written as \cite{Ricci}
\begin{equation} \label{ss}
ds^{2}=-\left(1-\frac{\mu}{r^{n-2}}\right)dt^{2}+\left(1-\frac{\mu}{r^{n-2}}\right)^{-1}dr^{2}+r^{2}d\Omega^{2}_{n-1}
\end{equation}
With
\begin{equation}
d\Omega^{2}_{n-1}=d\chi^{2}_{2}+\sin^{2}\chi_{2}d\chi_{3}^{2}+...+\prod^{n-1}_{m=2} \sin^{2}\chi_{m}d\chi_{n}^{2} 
\end{equation}
The retraction, deformation retract and folding of the $5D$ Schwarzchild space $Sc.$ has been discussed in \cite{nash} using the same method in \cite{holograph} which can be described as folllows: We compared the $5D$ Schwarzchild Ricci-flat metric ($n=4$ in (\ref{ss})) with the general form of the $5D$ flat metric $ds^{2}= -dx_{o}^{2}+ \sum_{i=1}^4 dx_{i}^{2}$ and obtained a set of coordinate transformations with a set of integration constants. Then we used Euler-Lagrange equations to study the geodesics in $Sc$ and obtained a set of equations with constants. We have found that: some types of the geodesic retractions of the $5D$ Schwarzchild space-time $Sc$ are circles $S_i \subset Sc$. So, a retraction $R : Sc. \rightarrow S_i ,\,i=1,2,3,4$ has been already defined and we need to define a homotopy between this retraction and the identity map on $Sc$ i.e.,  $g : Sc \times [0,1] \rightarrow Sc $ such that $g (s, 0) = s$ $\forall s\in Sc$, $f(s, 1) = R (s)$ $\forall R\in Sc$, $g(a, t) = a$ $\forall a\in S_i, t \in [0,1]$. For the retraction of $Sc$ into a geodesic $S_1 \subset Sc$, the homotopy is defined as $\xi^{S_1} : Sc \times [0,1] \rightarrow Sc$ where
\begin{equation}
\xi^{S_1}(m,s)=(1-s)~\xi(m,0)+s~\xi^{S_1}(m,1),~~~~~~~~~~~~\forall m \in Sc~ \&~ \forall s \in [0,1].
\end{equation}
With
\begin{equation}
\xi(m,0)=\left\{x_o, x_1, x_2, x_3, x_4\right\}~~\& ~~\xi^{S_1}(m,1)=\left\{x_o^{A=0}, x_1^{A=0}, x_2^{A=0}, x_3^{A=0}, x_4^{A=0}\right\}.
\end{equation}
The homotopies of the retraction of $Sc$ into geodesics $S_2,~S_3,S_4 \subset Sc$, are defined respectively as
\begin{equation}
\xi^{S_2}(m,s)=(1-s)~\xi(m,0)+s~\xi^{S_2}(m,1),~~~~~~~~~~~~~~~~~~\forall m \in Sc~ \&~ \forall s \in [0,1].
\end{equation}
\begin{equation}
\xi^{S_3}(m,s)=\cos \frac{\pi s}{2}~\xi(m,0)+\sin \frac{\pi s}{2}~\xi^{S_3}(m,1),~~~~~~~~~~~~\forall m \in Sc~ \&~ \forall s \in [0,1].
\end{equation}
\begin{equation}
\xi^{S_4}(m,s)=\cos \frac{\pi s}{2}~\xi(m,0)+\sin \frac{\pi s}{2}~\xi^{S_4}(m,1),~~~~~~~~~~~~\forall m \in Sc~ \&~ \forall s \in [0,1].
\end{equation}
With 
\begin{eqnarray*}
\xi^{S_2}(m,1)&=&\left\{x_o^{B=0}, x_1^{B=0}, x_2^{B=0}, x_3^{B=0}, x_4^{B=0}\right\}\\ .
\xi^{S_3}(m,1)&=&\left\{x_o^{\phi=0}, x_1^{\phi=0}, x_2^{\phi=0}, x_3^{\phi=0}, x_4^{\phi=0}\right\}\\ .
\xi^{S_4}(m,1)&=&\left\{x_o^{\theta=0}, x_1^{\theta=0}, x_2^{\theta=0}, x_3^{\theta=0}, x_4^{\theta=0}\right\} .
\end{eqnarray*}
Where the retractions $R : Sc. \rightarrow S_i ,\,i=1,2,3,4$ have been done setting $A=0$, $B=0$, $\phi=0$ and $\theta=0$. The superscripts denote the value of the coordinate when the constant vanishes, and the lengthy expressions for the coordinates $x_i^{A=0}$, $x_i^{B=0}$, $x_i^{\phi=0}$ and $x_i^{\theta=0}$ have been given in \cite{nash}. With the homotopies defined here, and the defined retraction in \cite{nash} we now have a deformation retract defined on $Sc$. Defining homotopies satisfying the three conditions in definition 3 proves the possible deformations of $Sc$.
\end{proof}

The physical interpretation of theorem (\ref{th2}) is that such blackholes can get continuously deformed, whether due to tidal effects or any other physical effects, and reduced to lower dimensions. As we have discussed in details in the introduction, this agrees with many quantum gravity approaches in which the effective dimension of space-time reduces to lower dimensions near the Planck scale. The spontaneous
dimensional reduction associated with black holes evaporation has also been discussed in \cite{vii}.

\section{Conclusion}
In summary, we have presented in this paper a new application of the homotopy and retraction theory in astrophysics and quantum gravity. A topological base has been introduced to the theories of deformations and dimensional reduction of black holes and wormholes. The deformation retract is, by definition, a homotopy between a retraction and the identity may. We have used this deformation retract to prove that wormholes/black holes can get deformed and reduced to lower dimensions. We have shown that while the homotopy provides a topological base to the theory of wormholes/compact objects deformations, the retraction provides a topological base to the theory of dimensional reduction which is believed to happen near Planck length. The proof developed here is restricted to objects (black holes, compact stars, wormholes) with Ricci-flat metrics and a more general method will be needed to deal with non Ricci-flat case. While most of the studies on deformations and dimensional reduction discuss the geometrical local side, the current study discusses the topological global side and provides a rigorous proof to their existence. We have discussed some basic mathematical aspects of homotopies such as existence, unicity/classes of transformations and well-definiteness. We have also discussed the effect of such deformations and retractions on the NEC violation in wormholes and deduced that such homotopic deformations should not affect the NEC violation. While the homotopy theory has been so useful in providing proofs and simplifying calculations in several areas of mathematical physics such as the study of defects in the ordered media of condensed matter physics, this paper represents another example in which the homotopy theory provides a proof for the existence of deformations in a certain class of astrophysical objects. However, a major difference between the current work and previous works is that the homotopy here is defined between a retraction and the identity map which means we have a deformation retract into subspaces defined on these Ricci-flat spaces. The topological retraction into subspaces represents a topological origin to the dimensional reduction in quantum gravity.  

\section*{Acknowledgment}
We are so grateful to the reviewer for his many valuable suggestions and comments that significantly
improved the paper.


\begin{thebibliography}{000}
\bibitem{holograph} Nasr Ahmed \& H. Rafat, Topological Origin of Holographic Principle: Application to Wormholes, Int. J. Geom. Meth. Mod. Phys. 15 (2018) 1850131, arXiv:gr-qc/1711.02553v2 . 
\bibitem{kinsey} L. C. Kinsey, Topology of surfaces (SpringerVerlag, New York, 1993).
\bibitem{cond} N. D. Mermin, The homotopy groups of condensed matter physics, J of Math Phys 19 (1978) 1457.
\bibitem{cond2} G. Toulouse (1980) A lecture on the topological theory of defects in ordered media: How the old theory was leading to paradoxes, and how their resolution comes within the larger frameworks of homotopy theory. In: P?kalski A., Przystawa J.A. (eds) Modern Trends in the Theory of Condensed Matter. Lecture Notes in Physics, vol 115. Springer, Berlin, Heidelberg.
\bibitem{basic} H. Lu and Jianwei Mei, Ricci-flat and charged wormholes in five dimensions, Phys. Lett. B666 (2008) 511. 
\bibitem{kenna} R.Kenna, Homotopy in statistical physics, Cond Mat Phys 2006, Vol. 9, No 2(46), pp. 283–304.
\bibitem{too} E. Bick, F.D. Steffen (Eds.), Topology and Geometry in Physics, Lect. Notes Phys. 659 (Springer,
Berlin Heidelberg 2005), DOI 10.1007/b100632.
\bibitem{ho} V. Turaev, Homotopy field theory in dimension and group-algebras, arXiv: math/9910010.
\bibitem{hoo} D. R. Finkelstein, Homotopy approach to quantum gravity, Int. J. Theor. Phys. 47 (2008) 534.
\bibitem{hooo} M. Benini, A. Schenkel \& L. Woike, Homotopy theory of algebraic quantum field theories, arXiv:1805.08795.
\bibitem{topo3} A. Chamblin, Some applications of differential topology in general relativity, J. Geom. Phys. 13 (1994) 357.
\bibitem{topo2}  Arvind Borde, Topology change in classical general relativity, arXiv:gr-qc/9406053.
\bibitem{geo} E. Bick and F.D. Steffen, Topology and geometry in physics (Springer, 2005).
\bibitem{einst} A. Einstein and N. Rosen, The particle problem in the general theory of relativity, Phys. Rev. 48 (1935), 73.
\bibitem{Ricci} M.M. Akbar and G W Gibbons, Ricci flat metrics with U(1) action and the Dirichlet boundary value problem in Riemannian quantum gravity and isoperimetric inequalities, Class.Quant.Grav. 20 (2003) 1787-1822 hep-th/0301026.
\bibitem{nash} Nasr Ahmed \& H. Rafat, Retract and folding of the 5D Schwarzchid field, Bulletin Math. Analysis and Applications 7 (2015) 10, arXiv: 1405.1057 
\bibitem{stien} H. A. Hamm and N. Mihalache, Deformation retracts of Stein spaces, Math. Ann. 308 (1997), 333.
\bibitem{wormhole} M. Visser, Lorentzian wormholes - from Einstein to Hawking, (AIP, New York, 1996).
\bibitem{retra} K. Borsuk, Theory of retracts, (PWN, 1967).
\bibitem{hatcher} A. Hatcher, Algebraic topology (Cambridge University Press, 2002).
\bibitem{gurl} N. Gurlebeck, Tidally distorted black holes. Springer Proceedings in Physics, 170 (2016).
\bibitem{tid1} T. Hinderer, B. D. Lackey, R. N. Lang, and J. S. Read, Tidal deformability of neutron stars with realistic equations of state
and their gravitational wave signatures in binary inspiral , Phys. Rev. D 81, 123016 (2010).
\bibitem{tid2} J. Vines, E. E. Flanagan, and T. Hinderer, Post-Newtonian tidal effects in the gravitational waveform from binary
inspirals, Phys. Rev. D 83, 084051 (2011).
\bibitem{tid3} F. Pannarale, L. Rezzolla, F. Ohme, and J. S. Read, Will black hole-neutron star binary inspirals tell us about the neutron
star equation of state? , Phys. Rev. D 84, 104017 (2011).
\bibitem{tid4} B. D. Lackey, K. Kyutoku, M. Shibata, P. R. Brady, and J. L. Friedman, Extracting equation of state parameters from
black hole-neutron star mergers. I. Nonspinning black holes, Phys. Rev. D 85, 044061 (2012).
\bibitem{tid5} B. D. Lackey, K. Kyutoku, M. Shibata, P. R. Brady, and J. L. Friedman, Extracting equation of state parameters from
black hole-neutron star mergers: Aligned-spin black holes and a preliminary waveform model, Phys. Rev. D 89, 043009.
(2014).
\bibitem{g1} S. Vacaru, Exact solutions in modified massive gravity and off-diagonal wormhole deformations, Eur. Phys. J. C 74 (2014) 2781.
\bibitem{g2} T. Gheorghiu, O. Vacaru, S. Vacaru, Off-diagonal deformations of Kerr black holes in Einstein and modified massive gravity and higher dimensions,Eur. Phys. J. C 74 (2014) 3152.
\bibitem{g3} T. Gheorghiu, Proceedings of the MG14 meeting on General Relativity, 2017.
\bibitem{g4} B. Chen \& L. C. Stein, Phys. Rev. D97 (2018) 084012.
\bibitem{brr1} J. Ovalle, in Gravitation and Astrophysics (ICGA9), (World Scientific, Singapore, 2010).
\bibitem{brr2} J. Ovalle, Extending the geometric deformation: New black hole solutions, Int. J. Mod. Phys. Conf. Ser. 41, 1660132 (2016).
\bibitem{de1} S. Chandrasekhar and F. Fermi, Problems of gravitational stability in the presence of a magnetic field, ApJ 118 (1953) 116.
\bibitem{de2} Ferraro V.C.A., On the equilibrium og magnetic stars, ApJ 119 (1954) 407.
\bibitem{de3} Cutler C., Gravitational waves from neutron stars with large toroidal B fields, Phys. Rev. D 66 (2002) 084025.
\bibitem{de4} Haskell B., Samuelsson S., Glampedakis K., and Andersson N., Modelling magnetically deformed neutron stars, MNRAS, 385 (2008) 531.
\bibitem{de5} A. Mastrano, P. D. Lasky and A. Melatos, Neutron star deformation due to multipolar magnetic fields, MNRAS 434 (2013)2.
\bibitem{vii} J. R. Mureika, Primordial Black Hole Evaporation and Spontaneous Dimensional Reduction, Phys. Lett. B 716, 171-175 (2012).
\bibitem{dim1} 2] L. A. Anchordoqui, et al, Searching for the Layered structure of space at the LHC, Phys. Rev. D 83, 114046 (2011) [arXiv:1012.1870 [hep-ph]].
\bibitem{dim2} D. Stojkovic, Vanishing dimensions: theory and phenomenology, Rom. J. Phys. 57, 210 (2012).
\bibitem{dim3} J. F. Nieves, Thermal field theory in a wire: applications of thermal field theory methods to the propagation of photons in a one-dimensional plasma, Int. J. Mod. Phys. A 26, 5387 (2011) [arXiv:1105.2546 [hep-ph]].
\bibitem{dim4} G. Calcagni, Geometry and field theory in multi-fractional space-time,JHEP 1201, 065 (2012) [arXiv:1107.5041 [hep-th]].
\bibitem{dim5} M. Rinaldi, Observational signatures of pre-inflationary and lower dimensional effective gravity, Class. Quant. Grav. 29 (2012) 085010 [arXiv:1011.0668 [astro-ph.CO]].
\bibitem{dim6} N. Arkani-Hamed, S. Dimopoulos and G. R. Dvali, Phenomenology, astrophysics and cosmology of theories with sub-millimeter dimensions and TeV scale quantum gravity, Phys. Rev. D 59, 086004
(1999) [arXiv:hep-ph/9807344]; N. Arkani-Hamed,
S. Dimopoulos and G. R. Dvali, The hierarchy problem and the new dimensions at a millimeter, Phys. Lett. B 429, 263 (1998) [arXiv:hep-ph/9803315].
\bibitem{dim7} L. Randall and R. Sundrum, An alternative to compactification, Phys. Rev. Lett. 83, 4690 (1999) [arXiv:hep-th/9906064];
L. Randall and R. Sundrum, A large mass hierarchy from a small extra dimensions, Phys. Rev. Lett. 83, 3370 (1999) [arXiv:hep-ph/9905221].
\bibitem{11}W. S. Massey, Algebric topology, an introduction (New York, 1967).
\bibitem{Witt} E. Witten, Quantum field theory and the Jones polynomial, Commun. Math. Phys. 121 (1989) 351.
\bibitem{ati} M. Atiyah, Topological quantum field theories, Inst. Hautes Etudes Sci. Publ. Math. 68 (1988) 175.
\bibitem{schwarz} A.S. Schwarz, The partition function of degenerate quadratic functional and Ray-Singer invariants, Lett. Math. Phys 2 (1978) 247.
\bibitem{h1} S. W. Hawking, Wormholes in spacetime, Phys. Rev. D37 (1988) 904.
\bibitem{h2} S. Coleman, Black holes as red herrings: Topological fluctuations and the loss of quantum coherenceNucl. Phys. B,307 (1988) 867.
\bibitem{bb1} G. Mandal, A. M. Sengupta, S. R. Wadia, Classical solutions of two-dimensional string theory, Mod. Phys. Lett. A 6 (1991) 1685.
\bibitem{bb2} E. Witten, String theory and black holes, Phys. Rev. D 44  (1991) 314.
\bibitem{bb3} J. R. Gott and M. Alpert, General relativity in a (2+1) dimensional space-time, Gen. Rel. Grav. 16 (1984) 243.
\bibitem{bb4} B. Reznik, Thermodynamics and evaporation of the (2+1) dimensional black hole, Phys. Rev. D 51 (1995) 1728.
\bibitem{bb5} G. Clement, Black hole mass and angular momentum in 2+1 gravity, Phys. Rev. D 68 (2002) 024032.
\bibitem{bb6} S. Carlip, 2+1-Dimensional Quantum Gravity, Cambridge University Press (2003).
\bibitem{bb7} L. Ortiz and M. P. Ryan, Quantum collapse of dust shells in 2+1 gravity, Gen. Rel. Grav. 39 (2007) 1087.
\bibitem{bb8} P. Collas, General relativity in two-and three-spacetimes, Am. J. Phys. 45 (1977) 833.
\bibitem{bb9} R. Jackiw, in Quantum Theory of Gravity, Adam Hilger, Bristol (1984).
\bibitem{bb10} A. Sikkema and R. B. Mann, Gravitation and cosmology in two-dimensions, Class. Quant. Grav. 8 (1991) 219.
\bibitem{bb11} R. B. Mann, Liouville black holes, Nucl. Phys. B418  (1994) 231.
\bibitem{bb12} D. Grumiller and R. Meyer, Ramifications of lineland, Turk. J. Phys. 30  (2006) 349.
\bibitem{bb13} S. Rajeev, Exact solutions of quantum gravity in 1+1 dimensions, Phys. Lett. B 113 (1982) 146.
\bibitem{bb14} J. R. Mureika and D. Stojkovic, Detecting vanishing dimensions via primordial gravitational wave astronomy, Phys. Rev. Lett. 106 (2011) 101101. 
\bibitem{bb16} G. Calcagni, Geometry and field theory in multi-fractional spacetime, JHEP 1201, 065 (2012).
\bibitem{ec1}  P. MartnMoruno and M. Visser, Semiclassical energy conditions for quantum vacuum states, JHEP 1309, 050 (2013) .
\bibitem{ec2} C. Barcelo and M. Visser, Twilight for energy conditions, Int. J. Mod. Phys. D 11, 1553 (2002).


\end{thebibliography}
\end{document}